\documentclass[lettersize,journal]{IEEEtran}
\usepackage{colortbl}
\usepackage{bm}

\usepackage{graphicx}  
\usepackage{url}       

\usepackage{amsmath}
\allowdisplaybreaks[4]
\usepackage{cite}

\usepackage{amsfonts,amssymb}
\usepackage{empheq}
\usepackage{stfloats}

\usepackage{cases}
\usepackage{algorithm}
\usepackage{multirow}
\usepackage{algorithmic}
\usepackage{epstopdf}
\usepackage{color}
\usepackage{bbm}
\usepackage{mathtools}
\usepackage{relsize}
\usepackage{caption}
\usepackage{subcaption}

\usepackage[square, comma, sort&compress, numbers]{natbib}

\makeatletter
\newcommand{\vvast}{\bBigg@{3.0}}
\newcommand{\vast}{\bBigg@{4}}
\newcommand{\Vast}{\bBigg@{4.5}}
\newcommand{\VVast}{\bBigg@{5}}
\newcommand{\VVVast}{\bBigg@{5.5}}
\newcommand{\VVVVast}{\bBigg@{6}}

\makeatother

\newtheorem{theorem}{{Theorem}}

\newcommand{\ls}[1]
{\dimen0=\fontdimen6\the\font
	\lineskip=#1\dimen0
	\advance\lineskip.5\fontdimen5\the\font
	\advance\lineskip-\dimen0
	\lineskiplimit=.9\lineskip
	\baselineskip=\lineskip
	\advance\baselineskip\dimen0
	\normallineskip\lineskip
	\normallineskiplimit\lineskiplimit
	\normalbaselineskip\baselineskip
	\ignorespaces
}

\begin{document}

	
	\title{\ls{1.0}{Fundamental for Delay and Reliability Guarantees for Emergency UAV}}

	\author{\IEEEauthorblockN{Wenchi Cheng\IEEEauthorrefmark{2}, Jingqing Wang\IEEEauthorrefmark{2},  Zhuohui Yao\IEEEauthorrefmark{2}, and Wei Zhang\IEEEauthorrefmark{3}}~\\[0.2cm]
		\IEEEauthorblockA{\IEEEauthorrefmark{2}State Key Laboratory of Integrated Services Networks\\
			Xidian University, Xi'an, China\\
			\IEEEauthorrefmark{3}School of Electrical Engineering and Telecommunications\\
			The University of New South Wales, Sydney, Australia\\
			E-mail: \{\emph{wccheng@xidian.edu.cn}, \emph{jqwangxd@xidian.edu.cn}, \emph{yaozhuohui@xidian.edu.cn}, \emph{w.zhang@unsw.edu.au}\}}
		
	}
	\maketitle
\begin{abstract}	
To support mission-critical services in emergency scenarios, wireless networks are required to provide stringent guarantees under massive Ultra-Reliable Low-Latency Communications (mURLLC) constraints. Distributed unmanned aerial vehicle (UAV)-based massive multiple-input multiple-output (MIMO) architectures have recently emerged as a promising solution for rapidly deployable emergency communication systems. However, how to fundamentally characterize and guarantee statistical quality-of-service (QoS) for such systems in the finite blocklength regime remains largely unexplored.
To overcome these challenges, in this paper we develop a fundamental analytical framework for delay and reliability bounded QoS guarantees in distributed UAV-based massive MIMO emergency networks under finite blocklength coding (FBC). 
By rigorously modeling the stochastic service process of distributed massive MIMO fading channels, we derive statistical characterizations the delay and error-rate bounded QoS exponents. 
We also establish QoS-driven controlling functions, including the $\epsilon$-effective capacity and the feasible QoS region.
Finally, the obtained simulation results validate and evaluate our developed modeling techniques and asymptotic formulations to support mURLLC.
\end{abstract}


	\section{Introduction}\label{sec:intro}
	\IEEEPARstart{N}{ext} generation wireless networks are expected to support a wide range of mission-critical services with extremely diverse and challenging requirements on delay, reliability, availability, etc., such as emergency response, disaster recovery, and autonomous unmanned systems. In such scenarios, communication failures may directly lead to catastrophic consequences, which fundamentally differentiates emergency wireless networks from conventional commercial cellular systems. Ensuring statistical quality-of-services (QoS) under highly uncertain and dynamically evolving environments therefore remains one of the most challenging problems in next-generation wireless communications.
	
	A defining characteristic of emergency scenarios is the partial or complete destruction of conventional terrestrial communication infrastructure, caused by natural disasters, accidents, or deliberate attacks. To rapidly restore connectivity and support time-critical sensing, communication, and control tasks, unmanned aerial vehicles (UAVs)~\cite{9762906} have emerged as a key enabler due to their rapid deployment capability, flexible mobility, and favorable propagation conditions. In particular, distributed UAV networks, where multiple UAVs cooperate to provide communication coverage, offer a resilient and scalable solution for emergency wireless systems~\cite{9919746}.

	However, as an anticipated new and dominating type of time/error-sensitive services over next generation wireless networks, emergency UAV-based networks must simultaneously support \textit{massive Ultra-Reliable Low-Latency Communications} (mURLLC)~\cite{10175166} to quantitatively design and evaluate stringent and diverse QoS performances, while simultaneously supporting massive connectivity, ultra-reliable, low-latency control links across a vast and chaotic disaster zone. The prevailing paradigm of using individual UAVs as simple aerial relays or base stations is fundamentally limited in its ability to meet the extreme demands of next-generation emergency response. This motivates the development of distributed and cell-free massive multiple-input multiple-output (MIMO)~\cite{10041787} architectures formed by UAV swarms, where each UAV is equipped with a small number of antennas and collaboratively serves ground users and unmanned devices. By exploiting distributed spatial diversity and cooperative transmission, such systems can achieve the fundamental benefits of massive MIMO, including channel hardening and favorable propagation~\cite{11003410}, while offering improved robustness against node failures and environmental uncertainties. As a result, distributed UAV-based massive MIMO constitutes a realistic and powerful communication architecture for emergency scenarios requiring massive connectivity and stringent QoS guarantees.	
	
	Beyond connectivity, emergency UAV-based networks are increasingly required to support delay- and reliability-critical traffics, such as real-time video streaming, situational awareness data collection, and coordination of autonomous unmanned systems. In these applications, delay and reliability requirements are tightly coupled, and communication QoS violations may directly compromise system stability. Therefore, understanding the fundamental delay and reliability limits of distributed UAV-based massive MIMO systems is essential for determining the feasibility boundaries of emergency wireless networks and for enabling feasible and efficient functionalities.

	Despite extensive research on statistical QoS theory~\cite{10945892,10879302} over massive MIMO systems, existing studies primarily focus on investigating QoS metrics and controlling functions in terms of the delay-bound violating probability under the assumption of infinite blocklength coding. 
	Such metrics are inadequate for emergency and mission-critical applications, where packets are short, decoding errors are non-negligible, and statistical guarantees on delay and reliability are required. 
	In particular, the stochastic behavior of distributed massive MIMO fading channels using finite blocklength coding (FBC)~\cite{yury2010,Yp2014}, and its impact on statistically bounded \textit{delay} and \textit{error-rate} QoS provisioning, remains insufficiently understood.
	Thus, it is crucial to identify and define the new fundamental statistical QoS metrics and characterizing their analytical relationships and connections, such as various \textit{QoS-exponent functions}, delay bound-violation probability, decoding-error probability, effective capacity, etc., especially considering distributed UAV-based massive MIMO models.

	To effectively overcome these challenges, in this paper we develop a fundamental analytical framework for statistical delay- and error-rate-bounded QoS guarantees in distributed UAV-based massive MIMO systems using FBC. 
	In particular, we develop FBC-based system models, including the wireless communication model and channel coding rate model.
	Then, we introduce novel modeling analytical frameworks for the fundamental asymptotics of the decoding error probability for our developed statistical delay and error-rate bounded QoS provisioning theory.
	Moreover, we establish the new QoS-driven $\epsilon$-effective capacity function and the feasible QoS region.
	Finally, we conduct the numerical analyses which validate and evaluate our developed modeling formulations in the finite blocklength regime.
	
The rest of this paper is organized as follows: Section~\ref{sec:sys} establishes FBC-based wireless networking models.
	Section~\ref{sec:qos} identifies and derives new statistical delay and error-rate bounded QoS metrics, the asymptomatic error-rate-QoS exponent, and concavity of the rate–reliability function.
Section~\ref{sec:EC} defines the $\epsilon$-effective capacity and feasible QoS region in the finite blocklength regime.
Section~\ref{sec:results} conducts the simulations and numerical analyses to validate and evaluate the system performances for our developed modeling frameworks.
The paper concludes with Section~\ref{sec:conclusion}.
	
	\section{The System Models}\label{sec:sys}
\begin{figure}
	\centering
	\includegraphics[scale=0.66]{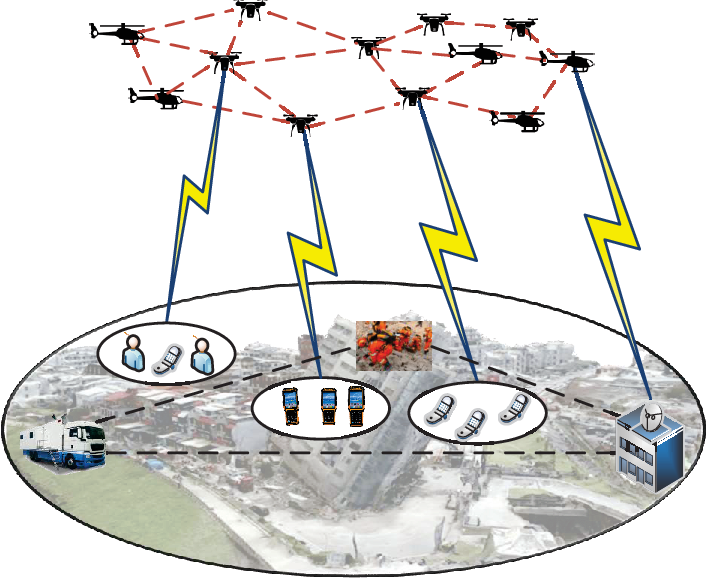}
	\caption{The distributed UAV-based massive MIMO architecture over emergency wireless networks.}
	\label{fig:1}
\end{figure}

We consider an emergency wireless communication scenario in which the terrestrial communication infrastructure is partially or fully unavailable due to natural disasters or unexpected incidents. To rapidly restore connectivity and support mission-critical services with stringent delay and reliability requirements, a swarm of unmanned aerial vehicles (UAVs) is deployed to form a distributed UAV-based massive multiple-input multiple-output (MIMO) network, as illustrated in Fig.~\ref{fig:1}. Each UAV is equipped with a small number of antennas, and the UAV swarm cooperatively serves multiple ground users or unmanned devices over a shared wireless medium.

From a communication-theoretic perspective, the distributed UAV antennas collectively form an equivalent massive MIMO system, in which cooperative transmission enables spatial diversity and array gains analogous to those of conventional centralized massive MIMO architectures. Let $N_{\text{T}}$ denote the total number of transmit antennas aggregated across the UAV swarm, and let $K_{\text{u}}$ denote the number of served ground devices generating delay- and reliability-critical traffic. Each ground device is assumed to be equipped with $N_{\text{R}}$ receive antennas.

\subsection{The Wireless Fading Channel Model}

Motivated by short-packet transmissions in mURLLC systems, we adopt a quasi-static block-fading channel model, under which the wireless channel remains constant over the duration of a codeword and varies independently from one codeword to another. This assumption is standard in finite blocklength analysis and captures the essential characteristics of delay-constrained communications.

Specifically, the channel between the distributed UAV antenna system and the $k$th ground device is modeled as a Rayleigh fading MIMO channel with channel matrix $\mathbf{H}_{k}\in\mathbb{C}^{N_{\text{R}}\times N_{\text{T}}}$. The channel matrix is expressed as
\begin{equation}
	\mathbf{H}_{k}=\sqrt{\xi_{k}}\mathbf{G}_{k}
\end{equation}
where $\xi{k}$ denotes the large-scale fading coefficient capturing the combined effects of path loss and shadowing between the UAV swarm and the $k$th device, and $\mathbf{G}_{k}$ represents the small-scale fading matrix whose entries are independent and identically distributed (i.i.d.) circularly symmetric complex Gaussian random variables with zero mean and unit variance.

Denote by $\mathbf{X}_{k}=\left[\bm{x}^{(1)}_{k},\dots, \bm{x}^{(n)}_{k}\right]\in{\cal X}_{n,N_{\text{T}}}$ and $\mathbf{Y}_{k}=\left[\bm{y}^{(1)}_{k},\dots, \bm{y}^{(n)}_{k}\right]$ the transmitted signal matrix and received signal matrix, respectively, from the distributed UAVs to a ground device, where $n$ is the blocklength and ${\cal X}_{n,N_{\text{T}}}$ is the transmit signal set satisfying the average power constraint.
The corresponding received signal matrix, denoted by $\mathbf{Y}_{k}\in\mathbb{C}^{N_{\text{R}}\times n}$, is given by
\begin{equation}\label{equation01}
	\mathbf{Y}_{k}=\sqrt{\frac{{\cal P}_{k}d_{k}^{-\tau_{k}}}{N_{\text{T}}}}\mathbf{H}_{k}\mathbf{X}_{k}+\mathbf{W}_{k}
\end{equation}
where ${\cal P}_{k}$ denotes the total transmit power allocated to user $k$, $d_{k}$ represents the distance between the UAV swarm and the $k$th ground device, $\tau_{k}$ denotes the path envelope, and $\mathbf{W}_{k}=\left[\bm{w}^{(1)}_{k},\dots,\bm{w}^{(n)}_{k}\right]$ is an additive white Gaussian noise (AWGN) matrix whose entries are i.i.d. $\mathcal{CN}(0,\sigma{k}^{2})$.
Under this model, the average signal-to-noise ratio (SNR), denoted by SNR$_{k}$, for user $k$ is given by
\begin{equation}\label{equation02a}
	\text{SNR}_{k}=\frac{{\cal P}_{k}d_{k}^{-\tau}}{\sigma_{k}^{2}}
\end{equation}
and the instantaneous channel SNR, denoted by $\gamma_{k}$, can be expressed as
\begin{align}\label{equation02}
	\gamma_{k}=&\frac{\text{SNR}_{k}}{N_{\text{T}}}\mathbf{H}_{k}\left(\mathbf{H}_{k}\right)^{H}
\end{align}
where $(\cdot)^{H}$ is the conjugate transpose.

\subsection{The Derivations for the Maximum Achievable Coding Rate}

The traditional Shannon's second theorem generally requires infinite blocklength for attaining the accurate approximation of the maximum achievable coding rate.
However, as noted above, Shannon's capacity formula cannot be applied when considering the limited bandwidth and \textit{stringent delay-bounded QoS requirements} in supporting mURLLC services in the finite blocklength regime for supporting the next generation mobile wireless networks.
Therefore, we consider an alternative solution by providing \textit{statistical} delay and error-rate bounded QoS guarantees through applying the FBC technique, where the QoS required by mURLLC services can be \textit{statistically} guaranteed with controlled small violation probabilities.
We can derive the \textit{normal approximation} of the \textit{maximum achievable coding rate}, denoted by $R^{*}\left(\gamma_{k}\right)$, between the UAV swarm and the $k$th ground device as follows~\cite{6802432}:
\begin{equation}\label{equation018}
	\epsilon_{k}\approx \mathbb{E}_{\gamma_{k}}\left[ {\cal Q}\left(\frac{C\left(\gamma_{k}\right)-R^{*}\left(\gamma_{k}\right)}{\sqrt{V\left(\gamma_{k}\right)/n}}\right)\right]
\end{equation}
where $\mathbb{E}_{\gamma_{k}}[\cdot]$ is the expectation over the SNR $\gamma_{k}$, ${\cal Q}(\cdot)$ is the \textit{$Q$}-function, $C\left(\gamma_{k}\right)$ and $V\left(\gamma_{k}\right)$ are the \textit{channel capacity} and \textit{channel dispersion}, respectively, which are given by the following equations, respectively:
\begin{equation}
	\begin{cases}	\vspace{1pt}
		C\left(\gamma_{k}\right)=
		\mathbb{E}_{\mathbf{H}_{k}}\left[{\cal I}\left(\mathbf{X}_{k};\mathbf{Y}_{k}|\mathbf{H}_{k}\right)\right]; \\
		V\left(\gamma_{k}\right)=\text{Var}\left[\bm{i}\left(\mathbf{X}_{k};\mathbf{Y}_{k}|\mathbf{H}_{k}\right)\right],
	\end{cases}
\end{equation}
where $\mathbb{E}_{\mathbf{H}_{k}}[\cdot]$ is the expectation operation with respect to $\mathbf{H}$ and $\text{Var}[\cdot]$ is the variance operation.
The normal approximation of the maximum achievable coding rate $R^{*}\left(\gamma_{k}\right)$ is obtained as the solution of Eq.~\eqref{equation018}.
We can derive the channel capacity $C\left(\gamma_{k}\right)$ as follows~\cite{telatar1999capacity}:
\begin{equation}
	C\left(\gamma_{k}\right)= \mathbb{E}_{\mathbf{H}_{k}}\left[\log_{2}\left[\det\left(\mathbf{I}_{N_\text{{R}}}+\frac{\text{SNR}_{k}}{N_{\text{T}}}\mathbf{H}_{k}\left(\mathbf{H}_{k}\right)^{H}\right)\right]\right].
\end{equation}
In high-end SNR region, the channel capacity, denoted by $C_{\text{High}}\left(\gamma_{k}\right)$, can be derived as follows:
\begin{align}\label{equation017}
	C_{\text{High}}\!\left(\gamma_{k}\right)\!=\!&\min\left\{N_\text{{T}},N_\text{{R}}\!\right\}\log_{2}\!\left(\!\frac{\text{SNR}_{k}}{N_{\text{T}}}\right)\!+\!\!\!\!\!\!\!\!\sum_{i=|N_\text{{T}}-N_\text{{R}}|+1}^{\max\left\{N_\text{{T}},N_\text{{R}}\right\}}\!\!\!\!\!\!\mathbb{E}\left[\log_{2}\left(\chi_{2i}^{2}\right)\right]\nonumber\\
	&+{\cal O}(1)
\end{align}
where $\chi_{2i}^{2}$ is the chi-square distribution with $2i$ degrees of freedom and ${\cal O}(\cdot)$ is the big O function.
The above normal approximation given in Eq.~\eqref{equation018} has been shown to be very tight for sufficiently large values of $n$~\cite{6802432}.

\section{The Fundamentals of Delay and Reliability QoS Guarantees}\label{sec:qos}

\subsection{The Statistical Delay and Error-Rate Bounded QoS Exponents} 

\subsubsection{The Statistical Delay-QoS Exponent}

\textit{Definition 1:} Based on the \textit{large deviations principle} (LDP), under sufficient conditions, the queueing process converges in distribution to a random variable $Q_{k}(\infty)$ such that
\begin{equation}\label{equation23}
	-\lim_{Q_{k,\text{th}}\rightarrow\infty}\frac{\log\left(\text{Pr}
		\left\{Q_{k}(\infty)>Q_{k,\text{th}}\right\}\right)}{Q_{k,\text{th}}}=\theta_{k}
\end{equation}
where $Q_{k,\text{th}}$ represents the overflow threshold for the queueing systems and $\theta_{k}$ $(\theta_{k}>0)$ is defined as the \textit{delay-QoS exponent} of queuing delay, which measures the exponential decay rate of the delay-bounded QoS violation probabilities. 		
$\hfill\blacksquare$

\subsubsection{The Error-Rate-QoS Exponent}

\textit{Definition 2:} Based on the LDP, the \textit{error-rate-QoS exponent}, denoted by $\vartheta_{k}$, measures the \textit{exponential decay rate} of the \textit{error-rate bounded QoS violation probabilities}, i.e., the decoding error probability $\epsilon_{k}$, which is defined as follows~\cite{gallager1968information}:
\begin{equation}\label{equation38}
	\vartheta_{k}\triangleq\lim\limits_{n\rightarrow \infty}-\frac{1}{n}\log(\epsilon_{k}),
\end{equation}
when the coding rate falls below the channel capacity, i.e. $R^{*}\left(\gamma_{k}\right)<C(\gamma_{k})$, where $R^{*}\left(\gamma_{k}\right)$ is given by Eq.~\eqref{equation018}.
Correspondingly, for a given coding rate $R^{*}\left(\gamma_{k}\right)<C(\gamma_{k})$, the decoding error probability $\epsilon_{k}$ vanishes exponentially as the packet length $n$ tends to infinity~\cite{gallager1968information,1055131}, i.e.,
\begin{equation}\label{equation039}
	\text{Pr}\left\{\widehat{W}_{k}\neq W_{k}\right\}\!=\!\epsilon_{k}\!\leq\!\exp\left\{-n \vartheta_{k}\right\}.
\end{equation}
$\hfill\blacksquare$

Based on the definition in~\cite{gallager1968information}, the \textit{error-rate-QoS exponent} $\vartheta_{k}$ is derived as follows:
\begin{equation}\label{equation040}
	\vartheta_{k}=\sup_{ \rho\in[0,1]}\left\{E_{k}(\rho)-\rho R^{*}\left(\gamma_{k}\right)\right\}
\end{equation}
where $\rho$ represents the Lagrange multiplier parameter, which is a real-valued number satisfying $\rho\in[0,1]$ and
\begin{align}\label{equation041}
	E_{k}(\rho)\!\triangleq& -\frac{1}{n}\log\! \bigg\{\mathbb{E}_{\gamma_{k}}\!\bigg[\int_{\mathbf{Y}_{k}}\!\bigg[\int_{\mathbf{X}_{k}}\!\!P_{\mathbf{X}_{k}}\left(\mathbf{X}_{k}\right)\! \nonumber\\
	&\quad \!\times\! \left[P_{\mathbf{Y}_{k}|\mathbf{X}_{k},\mathbf{H}_{k}}\!\left(\mathbf{Y}_{k}|\mathbf{X}_{k},\mathbf{H}_{k}\right)\!\right]^{\frac{1}{(1+\rho)}} \!\!d\mathbf{X}_{k}\bigg]^{(1+\rho)} d\mathbf{Y}_{k}\!\bigg]\!\bigg\}
\end{align}
where $P_{\mathbf{X}_{k}}\left(\mathbf{X}_{k}\right)$ is the PDF of the transmitted signal vector $\mathbf{X}_{k}$ between the UAV swarm and the $k$th ground device.
As a result, we need to derive the error-rate-QoS exponent function defined by Eq.~\eqref{equation040} for our developed FBC-based schemes for modeling and analyzing the error-rate bounded QoS metrics.
By substituting the conditional PDF $P_{\mathbf{Y}_{k}|\mathbf{X}_{k},\mathbf{H}_{k}}\left(\mathbf{Y}_{k}|\mathbf{X}_{k},\mathbf{H}_{k}\right)$ into Eq.~\eqref{equation041}, we can obtain the following equation~\cite{Gallager2009}:
\begin{align}\label{equation047}
	E_{k}(\rho)\!=\!-\!\frac{1}{n}\log\! \left\{\mathbb{E}_{\mathbf{H}_{k}}\!\left[\!\det\!\left(\!\mathbf{I}_{N_\text{{R}}}\!+\!\frac{\text{SNR}_{k}}{N_\text{{T}}(1+\rho)}\mathbf{H}_{k}\left(\mathbf{H}_{k}\right)^{H}\right)^{-n\rho}\right]\!\right\}.
\end{align}
Then, plugging Eq.~\eqref{equation047} into Eq.~\eqref{equation040}, we can derive the error-rate-QoS exponent $\vartheta_{k}$ as follows:
\begin{align}\label{equation046}
	\vartheta_{k}\!=&\!\sup_{ \rho\in[0,1]}\!\bigg\{\!-\frac{1}{n}\!\log\! \bigg\{\mathbb{E}_{\mathbf{H}_{k}}\!\bigg[\!\det\bigg(\mathbf{I}_{N_\text{{R}}}\!+\!\frac{\text{SNR}_{k}}{N_\text{{T}}(1\!+\!\rho)}
	\nonumber\\
	&\times \mathbf{H}_{k}\left(\mathbf{H}_{k}\right)^{H}\bigg)^{-n\rho}\bigg]\!\bigg\} -\rho R^{*}\left(\gamma_{k}\right)\bigg\}.
\end{align}
Moreover, considering the high-end SNR region, i.e., $\text{SNR}_{k}\rightarrow\infty$, we can rewrite Eq.~\eqref{equation047} as follows:
\begin{align}\label{equation075}
	&	E_{k}(\rho)=-\frac{1}{n}\log \left\{\mathbb{E}_{\mathbf{H}_{k}}\left[\det\left(\frac{\text{SNR}_{k}}{N_\text{{T}}(1+\rho)}\mathbf{H}_{k}\left(\mathbf{H}_{k}\right)^{H}\right)^{-n\rho}\right]\right\}
	\nonumber\\
	&\,\,=\!-\rho N_\text{{R}}\log(1\!+\!\rho)\!-\!\frac{1}{n}\!\log \!\left\{\!\mathbb{E}_{\lambda_{i}}\!\!\left[\!\left(\!\!	\prod_{i=1}^{\min\left\{N_\text{{T}},N_\text{{R}}\right\}}\!\!\!\lambda_{i}\text{SNR}_{k}\!\right)^{-n\rho\,}\!\right]\!\right\}\!.
\end{align}

Denote by $\rho^{*}$ the optimal Lagrange multiplier parameter that maximizes the error-rate-QoS exponent $\vartheta_{k}$.
The most challenging part in deriving the closed-form expression for the error-rate-QoS exponent $\vartheta_{k}$ defined in Eq.~\eqref{equation040} lies in solving for the closed-form solutions for $E_{k}(\rho)$ specified by Eq.~\eqref{equation047} and finding the optimizing value of $\rho^{*}$ that maximizes the error-rate-QoS exponent over massive MIMO based wireless fading channels.
Towards this end, we can obtain the error-rate-QoS exponent by deriving the approximate closed-form solution of $E_{k}(\rho)$ specified by Eq.~\eqref{equation047} in the high-end SNR region, which is summarized in the following theorem.

\begin{theorem}\label{theorem04}
	The \textit{error-rate-QoS exponent} $\vartheta_{k}$ in the high-end SNR region is approximately determined as follows:
	\begin{align}\label{theorem04_eq01}
		\vartheta_{k}
		&\!\approx\! \frac{\left[C(\gamma_{k})-R^{*}(\gamma_{k})\right]^{2}}{2\!\left\{2N_{\text{R}}\!+\!n\text{Var}_{\lambda_{i}}\left[\log\left(\prod\limits_{i=1}^{\min\left\{N_\text{{T}},N_\text{{R}}\right\}}\lambda_{i}\text{SNR}_{k}\right)\right]\right\}}.
	\end{align}
\end{theorem}

\begin{IEEEproof}
	The proof is omitted due to lack of space.	
\end{IEEEproof}

\indent\textit{Remarks on Theorem~\ref{theorem04}:} Theorem~\ref{theorem04} reveals that the error-rate-QoS exponent is \textit{a monotonically decreasing function with respect to the blocklength} $n$, which is consistent with the Definition~5 for the error-rate-QoS exponent $\vartheta_{k}$.
The approximate error-rate-QoS exponent plays an important role in \textit{fundamental-performance-limits} on statistical delay and error-rate bounded QoS.

 Furthermore, based on Eq.~\eqref{theorem04_eq01}, we can derive the error-rate-QoS exponent $\vartheta_{k}$ when the blocklength $n\rightarrow\infty$ as follows:
\begin{align}\label{theorem04_eq02}
	\lim\limits_{n\rightarrow\infty}\vartheta_{k}= 0.
\end{align}
This implies that when the blocklength $n\rightarrow\infty$, the error-rate-QoS exponent $\vartheta_{k}$, i.e., $\vartheta_{k}$, goes to zero, indicating that the system can tackle an arbitrary decoding error probability.

\subsection{Concavity of the Rate–Reliability Function}

\textit{Corollary~2:} For $e^{-n\vartheta_{k}}\in(0,1/2)$, the achievable rate $R^{*}\left(\gamma_{k}\right)$ is a strictly concave function of $\vartheta_{k}>0$.

\begin{IEEEproof}
	The proof is omitted due to lack of space.	
\end{IEEEproof}

\indent\textit{Remarks on Corollary~2:}  This concavity reveals a diminishing-return effect: achieving increasingly stringent reliability QoS requires disproportionately large sacrifices in transmission rate.

For large arguments, the inverse Gaussian tail function satisfies $Q^{-1}(e^{-n\vartheta_{k}})\approx \sqrt{2n\vartheta_{k}}$.
Thus, for sufficiently large $n\vartheta_{k}$, the rate–reliability function admits the following approximation:
\begin{equation}
	R^{*}\left(\gamma_{k}\right)\approx C\left(\gamma_{k}\right)-\sqrt{2V\left(\gamma_{k}\right)\vartheta_{k}}
\end{equation}
which coincides with the inverse of Gallager's random coding error exponent in the vicinity of capacity.

The concavity of the rate–reliability function shows that the marginal rate loss required to increase the error-rate bounded QoS exponent grows unbounded as $\vartheta_{k}$ increases, capturing the inherent inefficiency of ultra-reliable transmission.

\section{The Fundamental of $\epsilon$-Effective Capacity Under Delay and Reliability Guarantees}\label{sec:EC}

\subsection{The $\epsilon$-Effective Capacity for Delay and Error-Rate Bounded QoS}

\textit{Definition 3:} The \textit{$\epsilon$-effective capacity}, denoted by $EC\left(\theta_{k},\vartheta_{k}\right)$, is defined as the maximum constant arrival rate in the finite blocklength regime that a given service process can support the joint delay and error-rate bounded QoS requirements, which is formally expressed as follows:
\begin{align}\label{equation067}
	&EC\left(\theta_{k},\vartheta_{k}\right)\nonumber\\
	&\triangleq 
	\!-\frac{1}{n\theta_{k}}\!\log\bigg\{\!\mathbb{E}_{\gamma_{k}}\bigg[\epsilon_{k}+(1-\epsilon_{k})
	\exp\left\{-n\theta_{k}R^{*}\left(\gamma_{k}\right)\right\}\bigg] \bigg\} \nonumber\\
	&=\!-\frac{1}{n\theta_{k}}\!\log\bigg\{\!\mathbb{E}_{\gamma_{k}}\bigg[e^{-n \vartheta_{k}}\!+\!\left(1\!-\!e^{-n \vartheta_{k}}\!\right)
	e^{-n\theta_{k}R^{*}\left(\gamma_{k}\right)}\bigg]\bigg\}.
\end{align}
$\hfill\blacksquare$

To analyze the fundamental structure of the proposed $\epsilon$-effective capacity, we define the following log-moment generating function, denoted by $\Lambda\left(\theta_{k},\vartheta_{k}\right)$, of the service process:
\begin{equation}
	\Lambda\left(\theta_{k},\vartheta_{k}\right)\triangleq \!\log\bigg\{\!\mathbb{E}_{\gamma_{k}}\bigg[e^{-n \vartheta_{k}}\!+\!\left(1\!-\!e^{-n \vartheta_{k}}\!\right)
	e^{-n\theta_{k}R^{*}\left(\gamma_{k}\right)}\bigg]\bigg\}
\end{equation} 
and establish the fundamental convexity property of $\Lambda\left(\theta_{k},\vartheta_{k}\right)$ in terms of the delay and error-rate bounded QoS exponents.

\begin{theorem}\label{theorem08}
	For $e^{-n\vartheta_{k}}\in(0,1/2)$, the following claims hold for our developed performance modeling formulation in the finite blocklength regime:
	
	\underline{Claim 1.} The log-moment generating function $\Lambda\left(\theta_{k},\vartheta_{k}\right)$ is jointly convex in $\left(\theta_{k},\vartheta_{k}\right)$ over $\theta_{k},\vartheta_{k}>0$ in the finite-SNR region.
	
	\underline{Claim 2.} for a given $\theta_{k}>0$, the $\epsilon$-effective capacity is quasi-concave in $\vartheta_{k}$. 
\end{theorem}

\begin{IEEEproof}
	The proof is omitted due to lack of space.	
\end{IEEEproof}

\indent\textit{Remarks on Theorem~\ref{theorem08}:} The joint convexity structure implies a fundamental Pareto tradeoff between delay and error-rate bounded QoS. Increasing the delay-bounded QoS exponent $\theta_{k}$ penalizes service variability, while increasing the error-rate bounded QoS exponent reduces decoding failures at the cost of transmission rate. 

\subsection{The Convex Feasible QoS Region Under Joint Delay and Reliability Constraints}
According to Theorem~\ref{theorem08}, the log-moment generating function $\Lambda\left(\theta_{k},\vartheta_{k}\right)$ is jointly convex in 
$\left(\theta_{k},\vartheta_{k}\right)$ over $\mathbb{R}^{2}_{+}$. It is a fundamental result of convex analysis that the sublevel set of a convex function is convex.

\textit{Definition 9:} For any constant $u<0$, the feasible QoS region is defined as follows:
\begin{equation}\label{eq081}
	{\cal F}(u)\triangleq \left\{\left(\theta_{k},\vartheta_{k}\right)\in \mathbb{R}^{2}_{+}: \Lambda\left(\theta_{k},\vartheta_{k}\right)\leq u\right\} 
\end{equation}
Then, the feasible region ${\cal F}(u)$ is a convex set.
$\hfill\blacksquare$

From the perspective of $\epsilon$-effective capacity, Eq.~\eqref{eq081} can be rewritten as follows:
\begin{equation}
	EC\left(\theta_{k},\vartheta_{k}\right)\geq  -\frac{u}{n\theta_{k}},
\end{equation}
which characterizes the set of jointly admissible delay and error-rate bounded QoS exponents that can sustain a target effective service rate.
The feasible QoS region ${\cal F}(u)$ consists of all QoS exponent pairs under which the statistical delay violation probability decays at least exponentially with rate $\theta_{k}$, while the decoding error probability decays exponentially with rate $\vartheta_{k}$.
While the joint convexity characterizes the finite-SNR tradeoff between delay and reliability, the structure of the feasible QoS region undergoes fundamental changes across operating regimes.
In addition, the convexity of ${\cal F}(u)$ implies that time-sharing or convex combinations of two feasible QoS operating points remain feasible, enabling efficient joint optimization of delay and reliability constraints.

Then, we analyze the Pareto boundary and delay–reliability tradeoff.
The boundary of feasible QoS region ${\cal F}(u)$, defined by $\Lambda\left(\theta_{k},\vartheta_{k}\right)=u$, constitutes the Pareto-optimal delay–reliability tradeoff frontier. Along this boundary, any increase in the delay-bounded QoS exponent $\theta_{k}$ must be compensated by a reduction in the reliability-QoS exponent $\vartheta_{k}$, and vice versa. The strictly convex shape of this boundary reflects the diminishing returns in simultaneously tightening delay and reliability requirements.
As shown in Fig.~\ref{fig:002}, the shaded area represents jointly admissible delay and error-rate bounded QoS exponent pairs satisfying 
$\Lambda\left(\theta_{k},\vartheta_{k}\right)\leq u$. The solid curve corresponds to the Pareto-optimal boundary $\Lambda\left(\theta_{k},\vartheta_{k}\right)=u$, illustrating the fundamental tradeoff between delay and reliability stringency.

\begin{figure}
	\centering
	\includegraphics[scale=0.33]{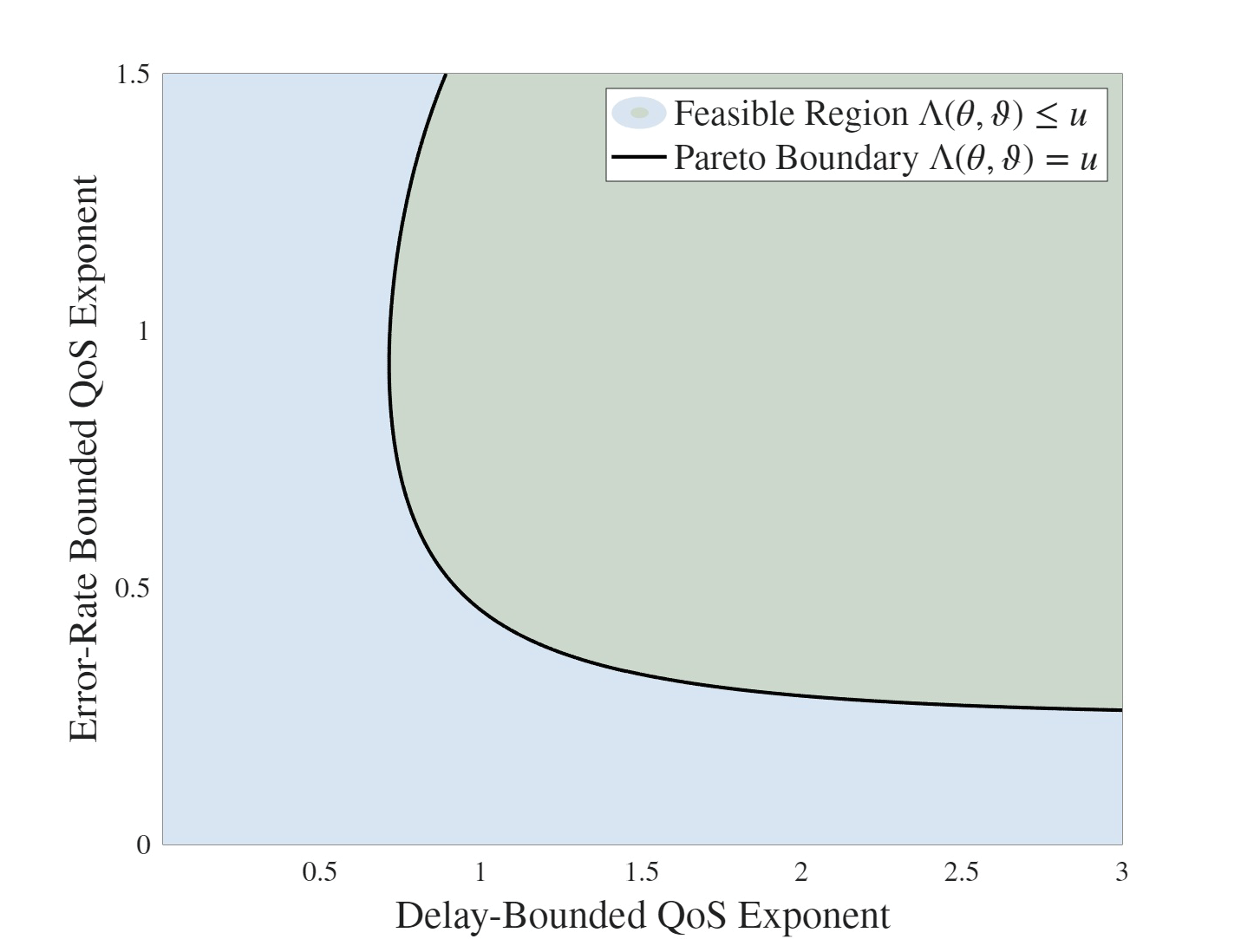}
	\caption{The convex feasible QoS region and Pareto boundary.}
	\label{fig:002}
\end{figure}

\subsection{Massive-MIMO Limit of the $\epsilon$-Effective Capacity}
In the high-SNR regime, the feasible QoS region expands monotonically due to the logarithmic growth of the achievable rate. However, the convexity of the feasible region and the Pareto-optimal delay–reliability boundary is preserved, indicating that higher SNR enlarges the admissible QoS space without altering its fundamental geometric structure.

\textit{Theorem~3:} In the high-SNR regime with a fixed number of transmit and receive antennas, \textbf{if} the achievable coding rate satisfies $R^{*}\left(\gamma_{k}\right) \propto\log\left(\text{SNR}_{k}\right)$, \textbf{then} for any fixed $\theta_{k}>0$, the log-moment generating function $\Lambda\left(\theta_{k},\vartheta_{k}\right)$ satisfies:
\begin{equation}
	\Lambda\left(\theta_{k},\vartheta_{k}\right)=-n\vartheta_{k}+{\cal O}(1), \qquad \text{SNR}\rightarrow\infty
\end{equation}
and the feasible QoS region asymptotically reduces to
\begin{equation}\label{eq082}
	{\cal F}(u)= \left\{\left(\theta_{k},\vartheta_{k}\right): \vartheta_{k}\leq -u/n\right\}.
\end{equation}

\begin{IEEEproof}
	The proof is omitted due to lack of space.	
\end{IEEEproof}

\indent\textit{Remarks on Theorem~3:}  Theorem~3 shows that the reliability dominates the service process, and the delay constraint becomes asymptotically non-binding.

\section{Performance Evaluations}\label{sec:results}

We provide numerical results to validate and evaluate our developed performance modelings over distributed UAV-based massive MIMO networks under statistical delay and reliability bounded QoS.
Consider a swarm of UAVs cooperatively serves a ground node in a cell-free manner. The cooperating UAVs are equipped with a single antenna and deployed at a fixed altitude. 
Throughout our simulations, we set the noise power $\sigma_{k}^{2}=0.1$ and the packet size of each data packet $M_{k,N_{\text{T}}}=10^{10}$ bits.

\begin{figure}
	\centering
	\includegraphics[scale=0.37]{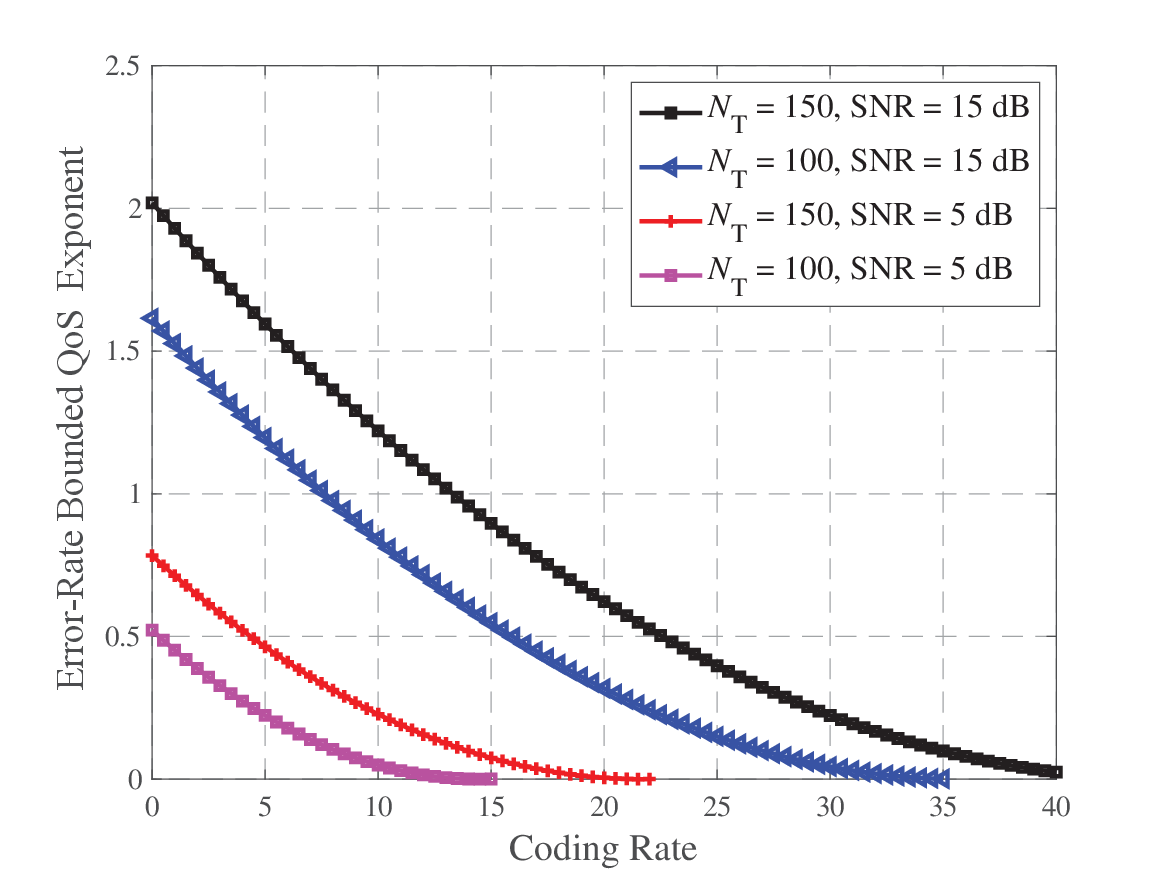}
	\caption{The error-rate-QoS exponent $\vartheta_{k}$ against coding rate $R\left(\gamma_{k}\right)$.}
	\label{fig:05}
\end{figure}

Figure~\ref{fig:05} plots the error-rate-QoS exponent $\vartheta_{k}$ as a function of the coding rate $R\left(\gamma_{k}\right)$ for our developed performance modeling formulations using FBC.		
Fig.~\ref{fig:05} shows that the error-rate-QoS exponent $\vartheta_{k}$ increases with the number of transmit antennas $N_{\text{T}}$ and also increases with the SNR.
In addition, Fig.~\ref{fig:05} also shows that the error-rate-QoS exponent $\vartheta_{k}$ decreases as the coding rate  $R\left(\gamma_{k}\right)$ when $R\left(\gamma_{k}\right)<C\left(\gamma_{k}\right)$.

Figure~\ref{fig:08} depicts the fundamental tradeoff between the optimal error-rate bounded QoS exponent and delay-bounded QoS exponent under FBC. Fig.~\ref{fig:08} shows the non-monotonic variation of the optimal $\vartheta_k$ with $\epsilon_{k}$. When $\epsilon_{k}$ is small, the system opts for high reliability, i.e., a large $\vartheta_k$. For moderate values of $\epsilon_{k}$, a moderate reduction in reliability is permitted to enhance the transmission rate. When $\epsilon_{k}$ becomes large, reliability must again be increased to reduce retransmission delays.

Figure~\ref{fig:010} depicts the $\epsilon$-effective capacity $EC\left(\theta_{k},\vartheta_{k}\right)$ as a function of both delay bounded QoS exponent $\theta_{k}$ and error-rate bounded QoS exponent $\vartheta_{k}$ in the finite blocklength regime.
Fig.~\ref{fig:010} shows the quasi-concavity of the $\epsilon$-effective capacity function. For a fixed $\theta_{k}$, the $\epsilon$-effective capacity exhibits a unimodal characteristic with respect to $\vartheta_{k}$, possessing a unique optimal value.
This implies that increasing the delay bounded QoS exponent $\theta_{k}$ generally leads to a reduction in the $\epsilon$-effective capacity. However, this decline can be alleviated to some extent by appropriately adjusting the error-rate bounded QoS exponent $\vartheta_{k}$.
Also, there exists an optimal path on the surface that forms a distinct ridge, which traces out the system's best possible operating points.

\begin{figure}
	\centering
	\includegraphics[scale=0.306]{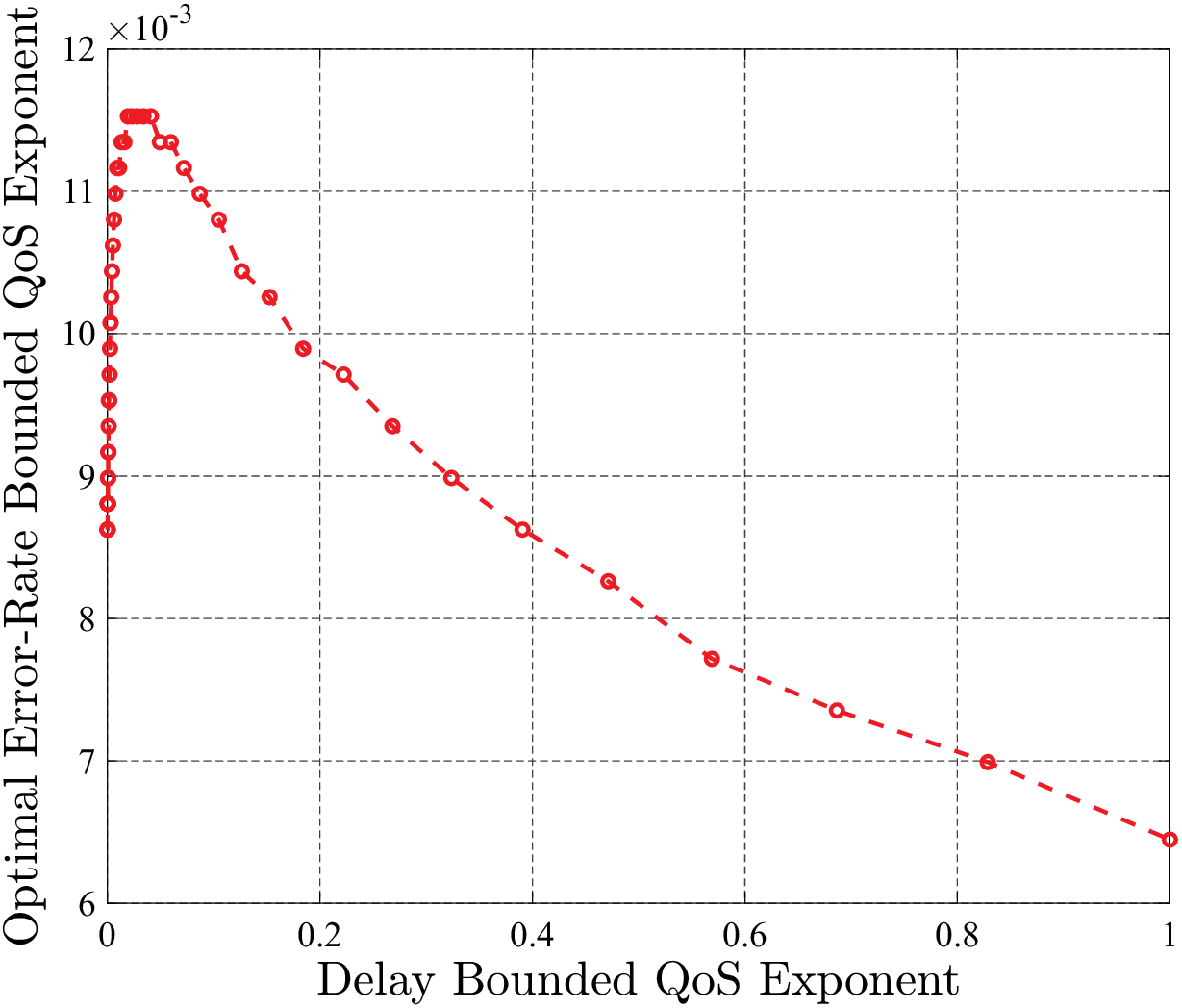}
	\caption{The optimal error-rate bounded QoS exponent against delay-bounded QoS exponent.}
	\label{fig:08}
\end{figure}

\begin{figure}
	\centering
	\includegraphics[scale=0.325]{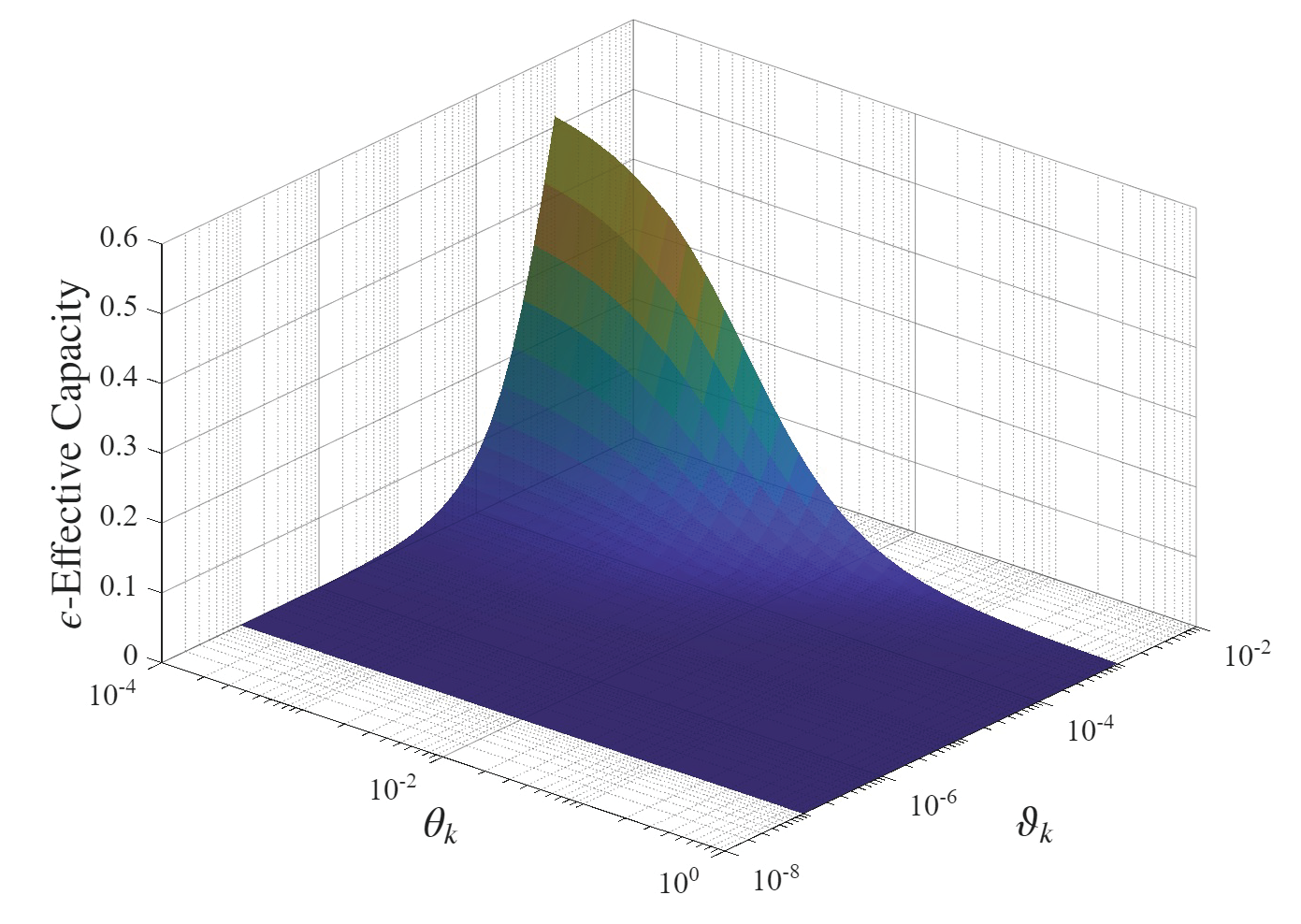}
	\caption{The $\epsilon$-effective capacity against delay bounded QoS exponent $\theta_{k}$ and error-rate bounded QoS exponent $\vartheta_{k}$.}
	\label{fig:010}
\end{figure}

\section{Conclusions}\label{sec:conclusion}
This paper investigated the fundamental limits of joint delay- and reliability-bounded QoS guarantees over distributed UAV-based massive MIMO networks under FBC. 
We developed a set of new statistical delay and error-rate bounded QoS metrics and controlling functions including the diverse {QoS exponents}, $\epsilon$-effective capacity, and the feasible QoS region under FBC.
Numerical results corroborated the theoretical analysis and illustrated the fundamental tradeoffs among blocklength, antenna dimensions, delay guarantees, and reliability constraints. 

		\nocite{*}
		\footnotesize
		\bibliographystyle{IEEEtran}
		\bibliography{myref1.bib}

	\end{document}